\documentclass[10pt,a4paper]{article}
\usepackage{amsmath,amsfonts,amssymb}
\usepackage{graphicx}
\usepackage{epstopdf}
\usepackage{cases}
\usepackage{mathrsfs}
\usepackage{parskip}
\usepackage{pst-all}

\usepackage{dsfont}

\newtheorem{theorem}{Theorem}

\newtheorem{definition}[theorem]{Definition}

\newenvironment{proof}[1][Proof]{\noindent\textbf{#1.} }{\ \rule{0.5em}{0.5em}}

\newcommand*{\degin}{\mathrm{deg}^-}
\newcommand*{\degout}{\mathrm{deg}^+}

\date{\today}
\author{Xiangying Chen\\Freie Universit\"{a}t Berlin\\ \texttt{chex48@zedat.fu-berlin.de}}

\title{Digraph Polynomials for Counting Cycles and Paths}

\begin{document}
	\maketitle
	\begin{abstract}
		Many polynomial invariants are defined on graphs for encoding the combinatorial information and researching them algebraically. In this paper, we introduce the cycle polynomial and the path polynomial of directed graphs for counting cycles and paths, respectively. They satisfy recurrence relations with respect to elementary edge or vertex operations. They are related to other polynomials and can also be generalized to the bivariate cycle polynomial, the bivariate path polynomial and the trivariate cycle-path polynomial. And a most general digraph polynomial satisfying such a linear recurrence relation is recursively defined and shown to be co-reducible to the trivariate cycle-path polynomial. We also give an explicit expression of this polynomial. 
	\end{abstract}
	
\section{Introduction}
	Many graph polynomials have been introduced and well studied over the years, they are shown to be effective on encoding, classifying and researching graph invariants as polynomials can be easily manipulated algebraically.  However digraph polynomials are presently less researched. A greater part of graph polynomials are generating functions for substructures in graphs. Most of them satisfy a linear recurrence relation with respect to elementary edge (e.g. \cite{averbouch2008most}) or vertex (e.g. \cite{tittmann2011enumeration}) operations. The relations between graph invariants can be researched by finding relations between graph polynomials. 
	In this present paper we define and research polynomial invariants for digraphs counting cycles and paths and research the class of digraph polynomials satisfying some linear recurrence relation. \\  
	In \cite{chung1995cover}, Chung and Graham introduced a bivariate digraph polynomial called the \emph{cover polynomial} which satisfies a Tutte-like deletion-contraction recurrence relation. It is one of the well-researched digraph polynomials. The research on digraph polynomials counting paths and cycles is motivated by the cover polynomial. It is defined recursively as
	\[
	C(D;x,y)=
	\begin{cases}
	C(D_{-e};x,y)+yC(D_{/e};x,y) & \textrm{if $e$ is a loop,}\\
	C(D_{-e};x,y)+C(D_{/e};x,y) & \textrm{if $e$ is not a loop,}\\
	\end{cases}
	\]
	and $C(E_n;x,y)=x^{\underline{n}}$ for arc-less digraph $E_n$.
	The combinatorial interpretation of $C(D;x,y)$ is 
	\[
	C(D;x,y)=\sum_{i,j}c_{i,j}(D)x^{\underline{i}}y^j,
	\]
	where $c_{i,j}(D)$ denotes the number of ways of disjointly covering all the vertices of $D$ with $i$ directed paths and $j$ directed cycles. 
	(Notice that isolated vertices are regarded as directed paths of length 0 by the cover polynomial and the following geometric cover polynomial. They will not be considered as directed paths in the polynomials defined in this paper.)\\
	The cover polynomial has been introduced as a digraph analogue of the Tutte polynomial. It is also a generalization of the rook polynomial. For the counting of cycle-path covers of a digraph, the ``normal" power can be used instead of the falling factorial. The \emph{geometric cover polynomial} introduced in \cite{d2000cycle} is the ordinary generating function for $c_{i,j}(D)$
	\[
	\widetilde{C}(D;x,y)=\sum_{i,j}c_{i,j}(D)x^iy^j.
	\]
	It satisfies the same recurrence relation as the cover polynomial, but the initial condition is $\widetilde{C}(E_n;x,y)=x^n$. In \cite{chung2016matrix}, these polynomials are generalized into matrix cover polynomial (for matrices, that is, multidigraphs or weighted digraphs) etc.\\
	This paper is structured as follows.
	In Section~\ref{sec31}, several digraph polynomials counting directed cycles and paths are introduced. These digraph polynomials satisfy arc deletion-contraction-extraction recurrence relations like the edge elimination polynomial \cite{averbouch2008most,averbouch2010extension} and vertex deletion-contraction recurrence relations. We give the relationships to their undirected versions and among them. 
	In Section~\ref{sec33}, we generalize the digraph polynomials counting cycles and paths and the geometric cover polynomial to the trivariate cycle-path polynomial. In Section~\ref{sec34}, applying the ideas of \cite{averbouch2008most}, the arc elimination polynomial is introduced, which is the most universal digraph polynomial satisfying linear recurrence relation with respect to deletion, contraction and extraction of arcs. We show that the arc elimination polynomial is co-reducible to the trivariate cycle-path polynomial. An explicit form of the arc elimination polynomial is given.
\section{The Cycle Polynomial and the Path Polynomial of Digraphs}\label{sec31}
In this paper, multidigraphs with loops are considered unless otherwise stated. The following arc operation for (multi-)digraphs will be used:
\begin{itemize}\setlength{\itemsep}{0pt}
	\item \emph{Arc deletion}. The graph obtained from $D$ by removing the arc $e$ is denoted by $D_{-e}$.
	\item \emph{Arc contraction}. If $e=(u,v)\in E$, $u\neq v$, $D_{/e}$ is defined as the digraph obtained from $D$ by unifying the two vertices $u$ and $v$ into a new vertex $w$, and removing exactly the arcs of the form $(u,x)$ or $(y,v)$ from $E$. If $e=(u,u)$, the vertex $u$ is also removed.
	\item \emph{Arc extraction}. For $e=(u,v)$, $D_{\dagger e}$ is defined as the digraph obtained from $D$ by removing $u$ and $v$ and their (or its if $u=v$) incident arcs.
	\item \emph{Arc addition}. The graph obtained from $D$ by adding the arc $(u,v),u,v\in V$ is denoted by $D_{+(u,v)}$.
\end{itemize}

A digraph $D$ on the vertex set $V(D)=\{1,\ldots,n\}$ can be represented as a matrix $A=(a_{ij})\in \mathbb{N}^{n\times n}$, where $a_{ij}$ is the number of arcs from vertex $i$ to vertex $j$. That is, $A$ is the adjacency matrix of $D$.\\
The digraph operations can be expressed as the matrix operations. Let $D=(V,E)$ be a digraph and $A(D)$ be the adjacency matrix of $D$. Without loss of generality, let $V=\{1,\ldots,n\}$. Then for $e=(i,j)\in E$:
\begin{itemize}\setlength{\itemsep}{0pt}
	\item $A(D_{-e})$ can be obtained from $A(D)$ by subtracting 1 from $a_{ij}$,
	\item $A(D_{/e})$ can be obtained from $A(D)$ by first exchanging row $i$ and row $j$ then deleting row $j$ and column $j$,
	\item $A(D_{\dagger e})$ can be obtained from $A(D)$ by deleting row $i$, row $j$, column $i$ and column $j$, and
	\item $A(D_{+(i,j)})$ can be obtained from $A(D)$ by adding 1 to $a_{ij}$.
\end{itemize}

\begin{figure}
	\includegraphics[width=350px]{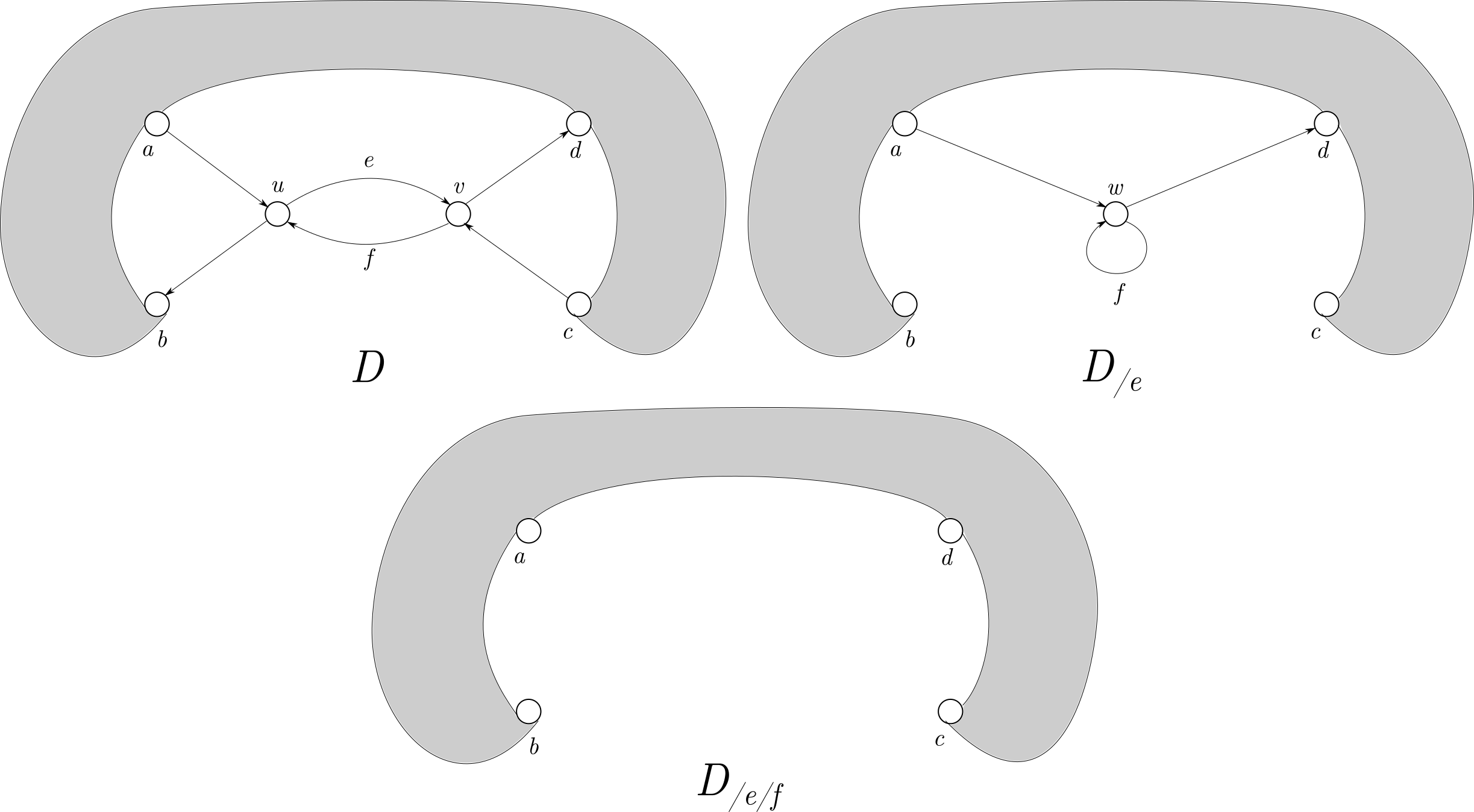}
	\caption{Arc contraction on a digraph}	
\end{figure}
Let $D=(V,E)$ be a digraph where multiple arcs and loops are allowed. The \emph{cycle polynomial} $\sigma(D)=\sigma(D;x)$ of the digraph $D$ is defined as 
\[
\sigma(D)=\sigma(D;x)=\sum_{k=1}^{|V|}c_k(D)x^k,
\]
where $c_k(D)$ denotes the number of directed cycles of length $k$ in $D$. Similarly, the \emph{path polynomial} of $D$ is defined as
\[
\pi(D)=\pi(D;x)=\sum_{k=1}^{|V|-1}p_k(D)x^k,
\]
where $p_k(D)$ denotes the number of directed paths of length $k$ in $D$.\\
The cycle polynomial and the path polynomial of digraphs satisfy respectively the following recurrence relations:
\begin{theorem}
	If $D=(V,E)$ is a digraph and $e\in E$ is an arc of $D$, then
	\[
	\sigma(D)=
	\begin{cases}
	\sigma(D_{-e})+x & \textrm{if $e$ is a loop,}\\
	\sigma(D_{-e})+x\sigma(D_{/e})-x\sigma(D_{\dagger e}) & \textrm{if $e$ is not a loop,}\\ & \textrm{and there are no loops on $u$ or $v$.}
	\end{cases}
	\]
\end{theorem}
\begin{proof}
	If $e$ is a loop, it is counted by $x$ and other cycles are counted by $\sigma(D_{-e})$. If $e$ is not a loop and there are no loops on $u$ or $v$, $\sigma(D_{-e})$ counts exactly directed cycles in $D$ without $e$. $D_{/e}$ contains exactly all cycles of $D$ containing $e$ with lengths decreased by 1, and all cycles of $D_{\dag e}$. Hence $x[\sigma(D_{/e})-\sigma(D_{\dag e})]$ counts exactly directed cycles of $D$ containing $e$.
\end{proof}

The recurrence for the path polynomial is similar. If $e$ is a loop it does not belong to any directed path and so can be deleted, but if $e$ is not a loop, $x$ must be added in order to count the directed path $e$. Then we have the following recurrence.
\begin{theorem}
	If $D=(V,E)$ is a digraph and $e\in E$ is an arc of $D$, then
	\[
	\pi(D)=
	\begin{cases}
	\pi(D_{-e}) & \textrm{if $e$ is a loop,}\\
	\pi(D_{-e})+x\pi(D_{/e})-x\pi(D_{\dagger e})+x & \textrm{if $e$ is not a loop.}
	\end{cases}
	\]
\end{theorem}
We can also transform a digraph into several digraphs in order to ensure that there is at most one arc between each pair of vertices. 
\begin{theorem}
	Let $D=(V,E)$ be a digraph and $u,v\in V$. Suppose that the arc $(u,v)$ has the multiplicity $n$ and the arc $(v,u)$ has the multiplicity $m$ in $E$. Let $D_3$ be the digraph obtained from $D$ by deleting all of the $n$ arcs $(u,v)$ and the $m$ arcs $(v,u)$, and let $D_1=D_{3+(u,v)}, D_2=D_{3+(v,u)}$, then
	\[
	\sigma(D)=n\sigma(D_1)+m\sigma(D_2)-(n+m-1)\sigma(D_3)+nmx^2,
	\]
	and
	\[
	\pi(D)=n\pi(D_1)+m\pi(D_2)-(n+m-1)\pi(D_3).
	\]
\end{theorem}
\begin{proof}
	The number of cycles in $D$ containing one of the arcs from $u$ to $v$ but no other arcs between $u$ and $v$ equals $m$ times the number of cycles in $D$ containing a fixed arc $e=(u,v)$ but no other arcs between $u$ and $v$, since $e$ can be replaced by any arc parallel to $e$ and form a different cycle. Therefore, these cycles can be counted by $m[\sigma(D_1)-\sigma(D_3)]$. Cycles in $D$ containing one of the arcs from $v$ to $u$ can be counted by $n[\sigma(D_2)-\sigma(D_3)]$. And cycles not containing any arcs between $u$ and $v$ can be counted by $\sigma(D_3)$. Cycles containing one arc $(u,v)$ and an arc $(v,u)$ is counted by $nmx^2$. We add these four terms together to obtain the reduction formula.\\
	The reduction for the path polynomial is analogous. The only difference is that the last term is not required.	
\end{proof}
\begin{figure}
	\includegraphics[width=350px]{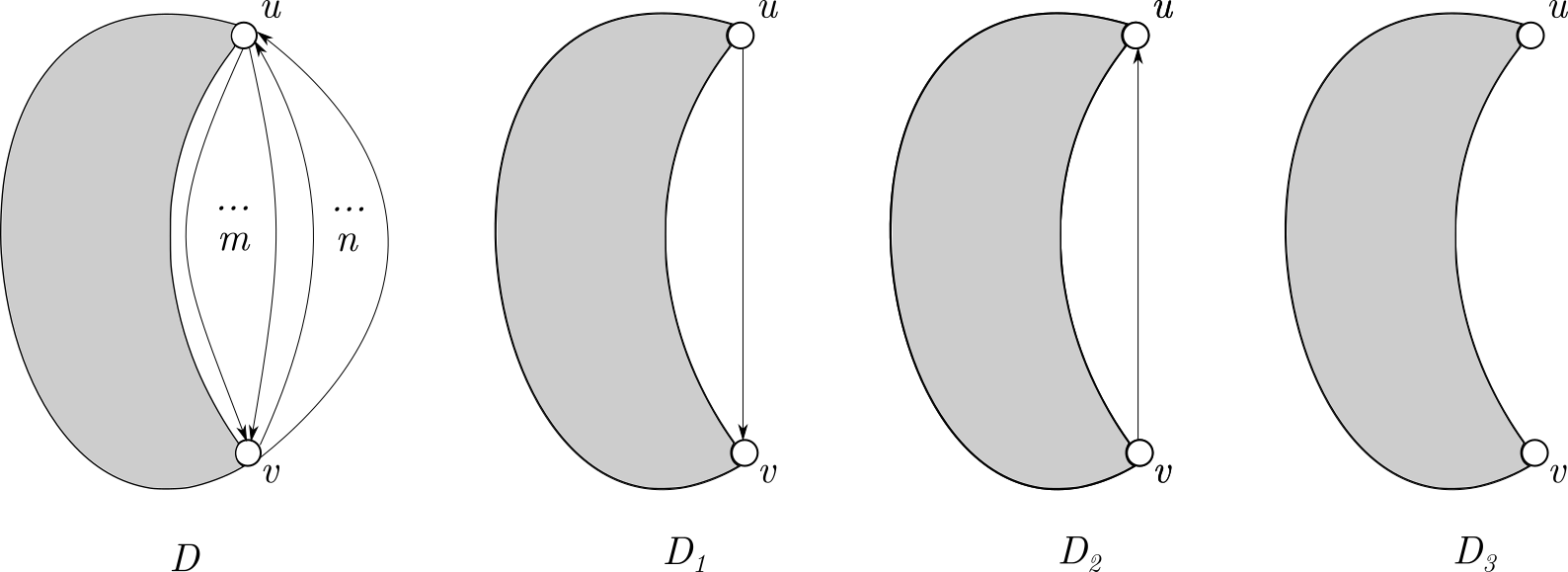}
	\caption{Simplification of parallel and anti-parallel arcs}		
\end{figure}

Now two vertex operations for digraphs need to be defined in order to state the vertex decomposition formulae. Given is a digraph $D=(V,E)$ and $v \in V$, the sets $N^+(v):= \{u\in V\big\vert (v,u)\in E\}$ and $N^-(v):= \{u\in V\big\vert (u,v)\in E\}$ are called the \emph{out-neighborhood} and the \emph{in-neighborhood} of $v$ in $D$, respectively.
\begin{itemize}\setlength{\itemsep}{0pt}
	\item Vertex deletion. The digraph obtained from $D$ by removing the vertex $v$ and all its incident arcs is denoted by $D_{-v}$.
	\item Vertex contraction. If the arcs incident with $v$ are not multiple, $D_{/e}$ is defined as the digraph obtained from $D_{-v}$ by adding the arcs of $N^-(v)\times N^+(v)$. For a multidigraph, the multiplicity of an added arc $(u,w)$ equals the multiplicity of $(u,v)$ times the multiplicity of $(v,w)$.
\end{itemize}
\begin{figure}
	\includegraphics[width=350px]{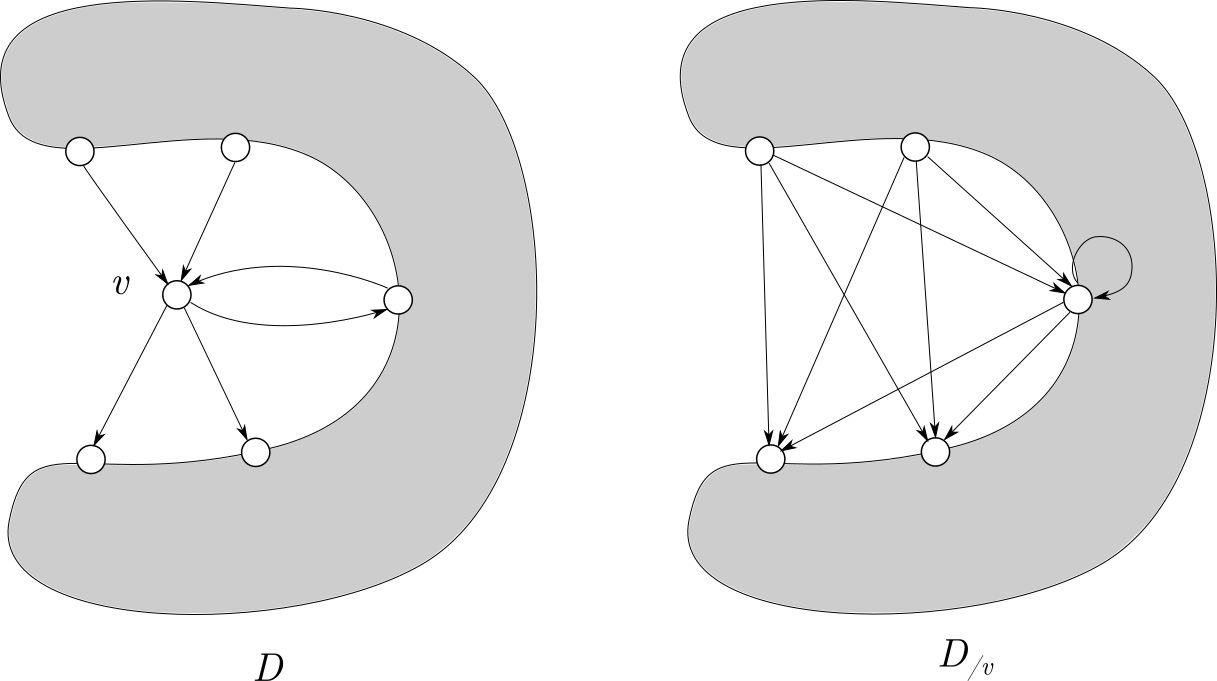}
	\caption{Vertex contraction on a digraph}		
\end{figure}
Given is a digraph $D=(V,E)$ and $v\in V$, $E^-(v)$ and $E^+(v)$ are defined to be the sets of arcs with head $v$ and tail $v$, respectively, that is, $E^+(v):=\{e=(v,u)\in E\}$ and $E^-(v):=\{e=(u,v)\in E\}$. We call $\degout(v):=|E^+(v)|$ the \emph{out-degree} and $\degin(v):=|E^-(v)|$ the \emph{in-degree} of $v$. The number of directed paths of length $k$ beginning with $v$ in $D$ is denoted by $p_k(D,v,+)$. Similarly, the number of directed paths of length $k$ ending with $v$ in $D$ is denoted by $p_k(D,v,-)$. We define the ordinary generating functions for $p_k(D,v,+)$ and $p_k(D,v,-)$:
\[
\pi_{v+}(D):=\sum_{k=1}^{|V|-1}p_k(D,v,+)x^k
\]
and
\[
\pi_{v-}(D):=\sum_{k=1}^{|V|-1}p_k(D,v,-)x^k.
\]
$\pi_{v+}(D)$ and $\pi_{v-}(D)$ satisfy the following decomposition formula:
\begin{theorem}
	Let $D=(V,E)$ be a digraph and $v\in V$, then
	\[
	\pi_{v+}(D)=\degout(v)\cdot x+x\sum_{u:(v,u)\in E^+(v)}\pi_{u+}(D_{-v}),
	\]
	and
	\[
	\pi_{v-}(D)=\degin(v)\cdot x+x\sum_{u:(u,v)\in E^-(v)}\pi_{u-}(D_{-v}).
	\]
\end{theorem}
\begin{proof}
	$\pi_{v+}(D)$ is the generating function for the number of paths of $D$ beginning at $v$. There are exactly $\degout(v)$ such paths of length 1. Each such path of length greater than 1 contains exactly one arc $(v,u)$ with tail $v$, and the remaining part of this path can be any path in $D_{-v}$ beginning at $u$. The proof of the second formula is analogous.
\end{proof}

Then we have the vertex decomposition formulae for $\sigma(D)$ and $\pi(D)$.
\begin{theorem}
	Let $D=(V,E)$ be a digraph and $v\in V$, then we have
	\[
	\sigma(D)=(1-x)\sigma(D_{-v})+x\sigma(D_{/v}),
	\]
	and
	\[
	\pi(D)=(1-x)\pi(D_{-v})+x\pi(D_{/v})+\pi_{v+}(D)+\pi_{v-}(D).
	\]
\end{theorem}
\begin{proof}
	$D_{/v}$ contains exactly cycles and paths of $D$ not containing $v$, and cycles and paths of $D$ containing $v$ but $v$ is neither source or sink of a path, with length decreased by 1.
\end{proof}

The decomposition formulae for digraphs are easier than that for graphs. There are also relationships between the digraph version and graph version of these polynomials. Let $\sigma(G)$ and $\pi(G)$ be ordinary generating functions for the undirected cycles and paths in an undirected graph $G$, respectively.
\begin{theorem}
	Let $D(G)$ denote the digraph obtained from the undirected graph $G$ by replacing each edge $\{u,v\}\in E$ by two oppositely oriented arcs $(u,v)$ and $(v,u)$. Then we have
	\[
	\pi(D(G))=2\pi(G),
	\]
	and
	\[
	\sigma(D(G))=2\sigma(G)+|E(G)|x^2.
	\]
\end{theorem}
\begin{proof}
	Each path or cycle of $G$ corresponds to two directed paths or cycles of different directions in $D(G)$. It is easy to see, directed cycles or paths of $G$ arising from different cycles or paths of $G$ are different, and all directed cycles or paths of $D(G)$ arise from corresponding cycles and paths of $G$ except the $|E|$ cycles consisting of two arcs arising from one edge of $G$.
\end{proof}

\section{Generalizations of Cycle and Path Polynomials for Digraphs}\label{sec33}
The next goal of this paper is to find the relationship between the (geometric) cover polynomial and our polynomials. Now we define the \emph{bivariate cycle polynomial} $\widehat{\sigma}(D)=\widehat{\sigma}(D;x,y)$ and the \emph{bivariate path polynomial} $\widehat{\pi}(D)=\widehat{\pi}(D;x,y)$ of a digraph $D$. Let $D=(V,E)$ be a digraph, let $kc(D)$ and $kp(D)$ denote the number of components of $D$ which are directed cycles and directed paths, respectively. We define
\[
\widehat{\sigma}(D)=\widehat{\sigma}(D;x,y)=\sum_{F}x^{|F|}y^{kc(D\langle F\rangle)},
\]
where the sum is over all subsets $F$ of $E$ that each component of the spanning subgraph $D\langle F\rangle$ is either a directed cycle or an isolated vertex. And we define
\[
\widehat{\pi}(D)=\widehat{\pi}(D;x,y)=\sum_{F}x^{|F|}y^{kp(D\langle F\rangle)},
\]
where the sum is over all subsets $F$ of $E$ that each component of the spanning subgraph $D\langle F\rangle$ is either a directed path or an isolated vertex. Obviously $\widehat{\sigma}(D)$ and $\widehat{\pi}(D)$ are multiplicative under components, and $\widehat{\sigma}(E_n)=\widehat{\pi}(E_n)=1$ for all $n\geq 0$. We have following recurrences for $\widehat{\sigma}(D)$ and $\widehat{\pi}(D)$:
\begin{theorem}
	\[
	\widehat{\sigma}(D)=
	\begin{cases}
	\widehat{\sigma}(D_{-e})+xy\widehat{\sigma}(D_{/e}) & \textrm{if $e$ is a loop,}\\
	\widehat{\sigma}(D_{-e})+x\widehat{\sigma}(D_{/e})-x\widehat{\sigma}(D_{\dag e}) & \textrm{otherwise.}
	\end{cases}
	\]
	\[
	\widehat{\pi}(D)=
	\begin{cases}
	\widehat{\pi}(D_{-e}) & \textrm{if $e$ is a loop,}\\
	\widehat{\pi}(D_{-e})+x\widehat{\pi}(D_{/e})+x(y-1)\widehat{\pi}(D_{\dag e}) & \textrm{otherwise.}
	\end{cases}
	\]
\end{theorem}
\begin{proof}
	Let $D=(V,E)$ be a digraph. 
	For $\widehat{\sigma}(D)$, we enumerate the arc subsets $F\subseteq E$ such that each component of the spanning subgraph $D\langle F\rangle$ is either a directed cycle or an isolated vertex. For each $e\in E$ there are two kinds of $F$: either $e\notin F$ or $e\in F$. \\
	If $e$ is a loop in $D$, the arc subsets $F$ of the first kind is counted by $\widehat{\sigma}(D_{-e})$. By the second kind, no other arcs in $F$ can be incident to the loop $e$, and the rest of $F$ corresponds to such an arc subset of $D_{\dagger e}=D_{/e}$. $e$ contributes one cycle of length $1$ and one arc to the polynomial. Thus, the second kind of $F$ is enumerated by $xy\widehat{\sigma}(D_{/e})$.\\
	If $e\in E$ is not a loop, the arc subsets $F$ not containing $e$ are counted by $\widehat{\sigma}(D_{-e})$. Consider now the digraph $D_{/e}$ and let $w$ be the new resulting vertex after contraction. Since all arcs with the same head or the same tail as $e$ are removed and the other arcs hold, each cycle of $D_{/e}$ containing $w$ corresponds to a cycle of $D$ containing $e$ and vice versa. The cycles of $D_{/e}$ not containing $w$ are identical to the cycles of $D_{\dagger e}$. However, $e$ contributes one arc to the polynomial. Thus the subsets $F$ of the second kind are enumerated by $x[\widehat{\sigma}(D_{/e})-\widehat{\sigma}(D_{\dag e})]$.\\
	The recurrence relation for $\widehat{\sigma}(D)$ is obtained by summing up these cases.\\
	Now consider $\widehat{\pi}(D)$. If $e$ is a loop, the spanning subgraphs of $D$ containing $e$ do not contribute to the polynomial. The spanning subgraphs of $D$ not containing $e$ are the spanning subgraphs of $D_{-e}$. That is, $\widehat{\pi}(D)=\widehat{\pi}(D_{-e})$ if $e$ is a loop.\\
	If $e$ is not a loop, in addition to the cases that contributed to the calculation of $\widehat{\sigma}(D)$ there is one more case: $e$ is the only arc of a component of the spanning subgraph. Any arc incident to $e$ cannot be in a spanning subgraph contributing to the polynomial, and $e$ contributes one arc and one directed path to the polynomial. Thus the spanning subgraphs containing $e$ as the only arc of a component, whose each component is either a directed path or an isolated vertex, are enumerated by $xy\widehat{\pi}(D_{\dagger e})$. Together with the other cases we obtain the recurrence relation.
\end{proof}

Furthermore, we can define the \emph{trivariate cycle-path polynomial} $\widehat{\sigma\pi}(D)=\widehat{\sigma\pi}(D;x,y,z)$ of a digraph $D$ counting all spanning subgraphs of $D$ whose components are either directed cycles or directed paths or isolated vertices: 
\[
\widehat{\sigma\pi}(D;x,y,z)=\!\!\!\!\sum_{\substack{F\subseteq E \\ \forall v\in V:\degout_{D\langle F\rangle}(v)\leq 1 \\ \forall v\in V:\degin_{D\langle F\rangle}(v)\leq 1}}\!\!\!\!x^{|F|}y^{kc(D\langle F\rangle)}z^{kp(D\langle F\rangle)}.
\]
Because of the same arguments as in the proof of the last theorem, we have the following recurrence relation for $\widehat{\sigma\pi}(D)$:
\begin{theorem}
	$\widehat{\sigma\pi}(D)=\widehat{\sigma\pi}(D;x,y,z)$ satisfies the following recurrence relation
	\[
	\widehat{\sigma\pi}(D)=
	\begin{cases}
	\widehat{\sigma\pi}(D_{-e})+xy\widehat{\sigma\pi}(D_{/e}) & \textrm{if $e$ is a loop,}\\
	\widehat{\sigma\pi}(D_{-e})+x\widehat{\sigma\pi}(D_{/e})+x(z-1)\widehat{\sigma\pi}(D_{\dag e}) & \textrm{otherwise.}
	\end{cases}
	\]
	And the initial condition is $\widehat{\sigma\pi}(E_n)=1$.
\end{theorem}

The following formulae follow direct from definition:
\[\sigma(D;x)=[y^1]\widehat{\sigma}(D;x,y),\]
\[\pi(D;x)=[y^1]\widehat{\pi}(D;x,y),\]
\[\widehat{\sigma}(D;x,y)=\widehat{\sigma\pi}(D;x,y,0),\]
\[\widehat{\pi}(D;x,y)=\widehat{\sigma\pi}(D;x,0,y).\]	
The geometric cover polynomial counts the number of cycle-path covers of a digraph. Since isolated vertices are regarded as directed paths of length 0, the number of paths in a cycle-path cover equals the number of vertices minus the number of arcs in this cover. We have the following relationship.
\begin{theorem}
	If $D=(V,E)$ is a digraph, then
	\[
	\widetilde{C}(D;x,y)=x^{|V|}\widehat{\sigma\pi}(D;\frac{1}{x},y,1).
	\]
\end{theorem}

\section{The Arc Elimination Polynomial for Digraphs}\label{sec34}
The digraph polynomials $C(D;x,y)$, $\sigma(D;x)$, $\pi(D;x)$, $\widehat{\sigma}(D;x,y)$, $\widehat{\pi}(D;x,y)$ and $\widehat{\sigma\pi}(D;x,y,z)$ satisfy certain linear recurrence relations with respect to deletion, contraction and extraction of an arc. In \cite{averbouch2008most}, Averbouch, Godlin and Makowsky introduced a most general undirected graph polynomial $\xi(G;x,y,z)$ satisfying an edge deletion-contraction-extraction linear recurrence relation, which generalizes the Tutte polynomial \cite{tutte1954contribution}, the matching polynomial \cite{farrell1979introduction} and the bivariate chromatic polynomial \cite{dohmen2003new}. The edge elimination polynomial is defined recursively as follows:
\begin{align*}
\begin{split}
&\xi(G;x,y,z)=\xi(G_{-e};x,y,z)+y\cdot\xi(G_{/e};x,y,z)+z\cdot\xi(G_{\dagger e};x,y,z),\\
&\xi(G_1\cup G_2;x,y,z)=\xi(G_1;x,y,z)\cdot\xi(G_2;x,y,z),\\
&\xi(E_1;x,y,z)=x,\\
&\xi(E_0;x,y,z)=1.
\end{split}
\end{align*}
In this section, we introduce the arc elimination polynomial for digraphs using the ideas of \cite{averbouch2011completeness,averbouch2008most}. 

\begin{theorem}\label{thm4.1}
	The digraph polynomial $\widehat{\xi}(D)=\widehat{\xi}(D;t,x,y,z)$ satisfying the recurrence relation
	\begin{align*}
	\begin{split}
	&\widehat{\xi}(D;t,x,y,z)=t\cdot\widehat{\xi}(D_{-e};t,x,y,z)+y\cdot\widehat{\xi}(D_{/e};t,x,y,z)+z\cdot\widehat{\xi}(D_{\dagger e};t,x,y,z),\\
	&\widehat{\xi}(G_1\cup G_2;x,y,z)=\widehat{\xi}(G_1;t,x,y,z)\cdot\widehat{\xi}(G_2;t,x,y,z),\\
	&\widehat{\xi}(E_1;t,x,y,z)=x,\\
	&\widehat{\xi}(E_0;t,x,y,z)=1
	\end{split}
	\end{align*}
	is well-defined iff $t=1$ or $y=z=0$. In the latter case, $\widehat{\xi}(D)=t^{|E(D)|}x^{|V(D)|}$.
\end{theorem}
\begin{proof}
	First, we prove that $t=1$ or $y=z=0$ is the necessary condition for the well-definedness of $\widehat{\xi}(D)$. First consider two arcs $e=(u,v)$, $f=(v,w)$ in $E(D)$, where $u$, $v$ and $w$ are different vertices. In order to be well-defined, $\widehat{\xi}(D)$ must return the same value when the decomposition is applied first to the arc $e$ and then to the arc $f$, as well as when it is applied first to $f$ then to $e$.\\
	Applying decomposition first to $e$ then to $f$, we have
	\begin{align*}
	\begin{split}
	\widehat{\xi}(D)&=t\cdot\widehat{\xi}(D_{-e})+y\cdot\widehat{\xi}(D_{/e})+z\cdot\widehat{\xi}(D_{\dagger e})\\
	&=t^2\cdot\widehat{\xi}(D_{-e-f})+ty\cdot\widehat{\xi}(D_{-e/f})+tz\cdot\widehat{\xi}(D_{-e\dagger f})\\
	&\hspace{10pt}+ty\cdot\widehat{\xi}(D_{/e-f})+y^2\cdot\widehat{\xi}(D_{/e/f})+yz\cdot\widehat{\xi}(D_{/e\dagger f})+z\cdot\widehat{\xi}(D_{\dagger e})\\
	&=t^2\cdot\widehat{\xi}(D_{-e-f})+ty\cdot\widehat{\xi}(D_{-e/f})+tz\cdot\widehat{\xi}(D_{\dagger f})\\
	&\hspace{10pt}+ty\cdot\widehat{\xi}(D_{-f/e})+y^2\cdot\widehat{\xi}(D_{/e/f})+yz\cdot\widehat{\xi}(D_{\dagger e\dagger f})+z\cdot\widehat{\xi}(D_{\dagger e}),
	\end{split}
	\end{align*}
	and first on $f$ then on $e$, we have
	\begin{align*}
	\begin{split}
	\widehat{\xi}(D)&=t\cdot\widehat{\xi}(D_{-f})+y\cdot\widehat{\xi}(D_{/f})+z\cdot\widehat{\xi}(D_{\dagger f})\\
	&=t^2\cdot\widehat{\xi}(D_{-f-e})+ty\cdot\widehat{\xi}(D_{-f/e})+tz\cdot\widehat{\xi}(D_{-f\dagger e})\\
	&\hspace{10pt}+ty\cdot\widehat{\xi}(D_{/f-e})+y^2\cdot\widehat{\xi}(D_{/f/e})+yz\cdot\widehat{\xi}(D_{/f\dagger e})+z\cdot\widehat{\xi}(D_{\dagger f})\\
	&=t^2\cdot\widehat{\xi}(D_{-e-f})+ty\cdot\widehat{\xi}(D_{-f/e})+tz\cdot\widehat{\xi}(D_{\dagger e})\\
	&\hspace{10pt}+ty\cdot\widehat{\xi}(D_{-e/f})+y^2\cdot\widehat{\xi}(D_{/e/f})+yz\cdot\widehat{\xi}(D_{\dagger e\dagger f})+z\cdot\widehat{\xi}(D_{\dagger f}).
	\end{split}
	\end{align*}
	They must coincide because of the well-definedness of $\widehat{\xi}(D)$. We have
	\[
	tz\cdot\widehat{\xi}(D_{\dagger f})+z\cdot\widehat{\xi}(D_{\dagger e})=tz\cdot\widehat{\xi}(D_{\dagger e})+z\cdot\widehat{\xi}(D_{\dagger f}),
	\]
	that is,
	\[
	(t-1)z\cdot\widehat{\xi}(D_{\dagger e})=(t-1)z\cdot\widehat{\xi}(D_{\dagger f}),
	\]
	which leads to $t=1$ or $z=0$ or $\widehat{\xi}(D_{\dagger e})=\widehat{\xi}(D_{\dagger f})$.\\
	Consider the latter case. Let $D$ be a digraph and $v$ an arbitrary vertex of $D$. Let $D'$ be the digraph obtained from $D$ by adding two vertices $u,w\notin V(D)$ and two arcs $e=(v,u),f=(u,w)$ to $D$. Applying extraction on $e$ and $f$, we have $D'_{\dagger e}=D_{-v}\cup K_1$ and $D'_{\dagger f}=D$. Since $\widehat{\xi}(D'_{\dagger e})=\widehat{\xi}(D'_{\dagger f})$, we have $\widehat{\xi}(D_{-v}\cup E_1)=\widehat{\xi}(D)$ for any vertices $v\in V(D)$. Applying this on every vertex of $D$, we get a trivial polynomial $\widehat{\xi}(D)=\widehat{\xi}(E_{|V(D)|})=x^{|V(D)|}$. This is a evaluation of $\widehat{\xi}(D)$ at $t=1$, $y=z=0$. That is, the third case is contained in the first case.\\
	Consider now the second case $\widehat{\xi}(D;t,x,y,0)$ and two arcs $e=(u,v)$, $f=(w,v)$ in $E(D)$, where $u$, $v$ and $w$ are different. Applying decomposition first on $e$ then on $f$ we get
	\begin{align*}
	\begin{split}
	\widehat{\xi}(D;t,x,y,0)&=t\cdot\widehat{\xi}(D_{-e};t,x,y,0)+y\cdot\widehat{\xi}(D_{/e};t,x,y,0)\\
	&=t^2\cdot\widehat{\xi}(D_{-e-f};t,x,y,0)+ty\cdot\widehat{\xi}(D_{-e/f};t,x,y,0)+y\cdot\widehat{\xi}(D_{/e};t,x,y,0)\\
	&=t^2\cdot\widehat{\xi}(D_{-e-f};t,x,y,0)+ty\cdot\widehat{\xi}(D_{/f};t,x,y,0)+y\cdot\widehat{\xi}(D_{/e};t,x,y,0)
	\end{split}
	\end{align*} 
	Applying decomposition first on $f$ then on $e$, we get
	\begin{align*}
	\begin{split}
	\widehat{\xi}(D;t,x,y,0)&=t\cdot\widehat{\xi}(D_{-f};t,x,y,0)+y\cdot\widehat{\xi}(D_{/f};t,x,y,0)\\
	&=t^2\cdot\widehat{\xi}(D_{-f-e};t,x,y,0)+ty\cdot\widehat{\xi}(D_{-f/e};t,x,y,0)+y\cdot\widehat{\xi}(D_{/f};t,x,y,0)\\
	&=t^2\cdot\widehat{\xi}(D_{-e-f};t,x,y,0)+ty\cdot\widehat{\xi}(D_{/e};t,x,y,0)+y\cdot\widehat{\xi}(D_{/f};t,x,y,0).
	\end{split}
	\end{align*} 
	From the coincidence of two results we have
	\[
	(t-1)y\cdot\widehat{\xi}(D_{/e};t,x,y,0)=(t-1)y\cdot\widehat{\xi}(D_{/f};t,x,y,0).
	\]
	The well-definedness implies that $t=1$ or $y=0$ or $\widehat{\xi}(D_{/e})=\widehat{\xi}(D_{/f})$. If $y=0$, then $\widehat{\xi}(D)=t\cdot\widehat{\xi}(D_{-e})$ and $\widehat{\xi}(E_n)=x^n$, which yields immediately that $\widehat{\xi}(D)=t^{|E(D)|}x^{|V(D)|}$. If $\widehat{\xi}(D_{/e})=\widehat{\xi}(D_{/f})$, given any digraph $D$ and let $v$ be any vertex of $D$. Let $D'$ be the digraph obtained from $D$ by adding two vertices $u,w\notin V(D)$ and two arcs $e=(v,u)$, $f=(w,u)$ to $D$. Applying the contraction on $e$ and $f$, we have $D'/f=D\cup E_1$ and $D'/f=D_{-E^+(v)}\cup E_1$. $\widehat{\xi}(D'_{/e})=\widehat{\xi}(D'_{/f})$ implies
	\[
	\widehat{\xi}(D\cup E_1)=\widehat{\xi}(D_{-E^+(v)}\cup E_1).
	\]
	From the definition of $\widehat{\xi}(D)$ we have
	\[
	\widehat{\xi}(D)=\widehat{\xi}(D_{-E^+(v)})
	\]
	for any digraph $D$ and any vertex $v$ in $D$. Applying $D_{-E^+(v)}$ on every vertex of $D$, we have $\widehat{\xi}(D)=\widehat{\xi}(E_{|V(D)|})=x^{|V(D)|}$. In this case, it is the trivial polynomial $\widehat{\xi}(D;1,x,0,0)=x^{|V(D)|}$.\\
	So far, we proved that the necessary condition is $t=1$ or $y=z=0$. The well-definedness in case $y=z=0$ is ensured by the explicit formula $\widehat{\xi}(D)=t^{|E(D)|}x^{|V(D)|}$. Consider the case $t=1$. We denote this possible polynomial by the notation of edge elimination polynomial:
	\[\xi(D;x,y,z):=\widehat{\xi}(D;1,x,y,z).\]
	Then we should prove the well-definedness of $\xi(D;x,y,z)$, that is, the result is independent of the order of decomposition steps.\\
	The distributivity of multiplication implies that elimination of an arc is exchangeable with decomposition of disjoint union. Hence, we can assume that the disjoint union decomposition steps are applied only on empty graphs, and only consider the order of decomposition of arcs.\\
	We shall consider only the linear order over arcs rather than decomposition steps. Such an order uniquely determines the decomposition process, if by convention, we just skip the steps of removing arcs that have been already removed by the proceeding steps. It is enough to show that successively decomposed arcs can be swapped. For two arcs $e,f\in E(D)$ there are 11 possible cases as shown in Figure~\ref{fig3.4}.
	\begin{figure}\label{fig3.4}
		\includegraphics[width=350px]{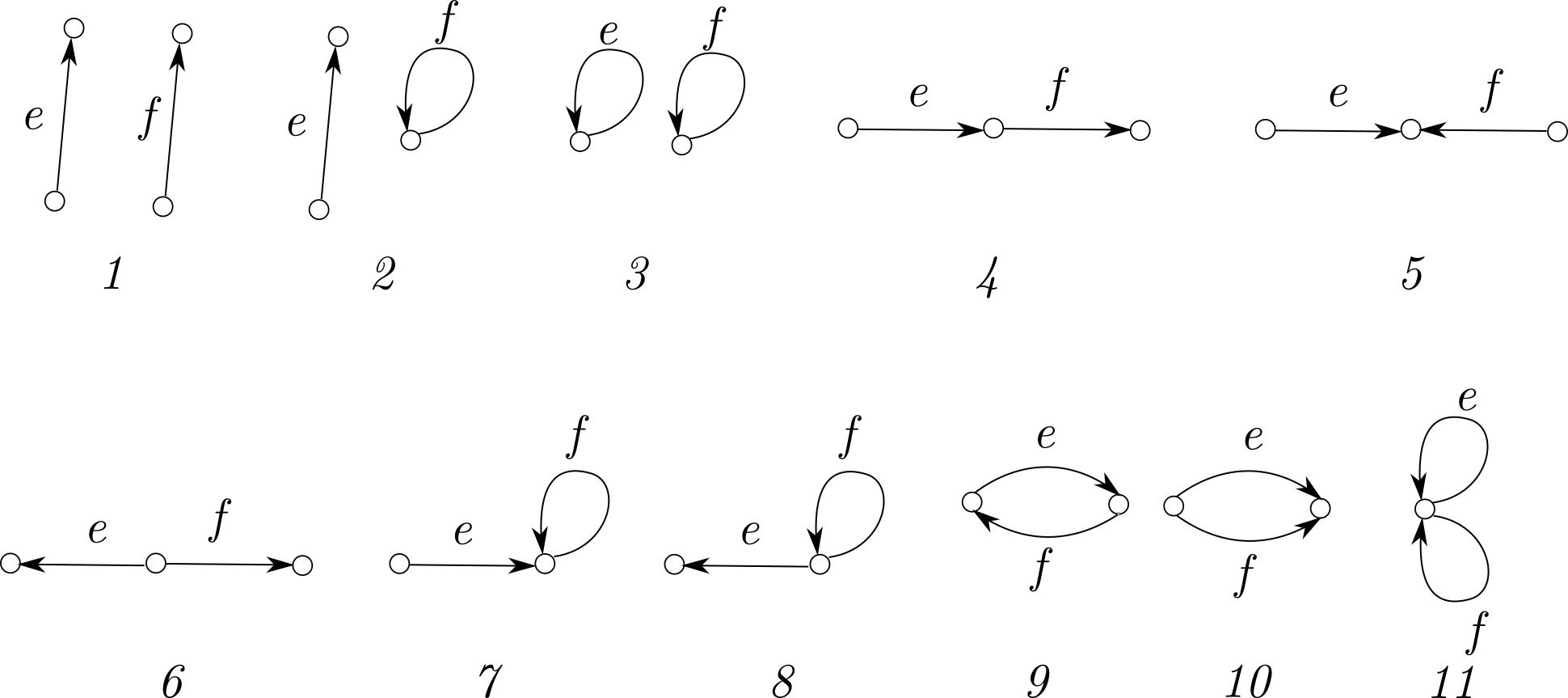}
		\caption{11 cases for two arcs $e$ and $f$}		
	\end{figure}
	In the case 1-3, the arc elimination operations are independent and hence commutative. In case 4 and case 5 the exchangeablility of elimination order of $e$ and $f$ are already showed. The case 6 is the same as case 5. In the case 7 and 8 we decompose first on $e$ then on $f$ and have
	\begin{align*}
	\begin{split}
	\xi(D;x,y,z)&=\xi(D_{-e};x,y,z)+y\cdot\xi(D_{/e};x,y,z)+z\cdot\xi(D_{\dagger e};x,y,z)\\
	&=\xi(D_{-e-f};x,y,z)+y\cdot\xi(D_{-e/f};x,y,z)+z\cdot\xi(D_{-e\dagger f};x,y,z)\\
	&\hspace{10pt}+y\cdot\xi(D_{/e};x,y,z)+z\cdot\xi(D_{\dagger e};x,y,z)\\
	&=\xi(D_{-e-f};x,y,z)+(y+z)\cdot\xi(D_{\dagger f};x,y,z)+y\cdot\xi(D_{/e};x,y,z)\\
	&\hspace{10pt}+z\cdot\xi(D_{\dagger e};x,y,z).
	\end{split}
	\end{align*}
	Applying decomposition first on $f$ then on $e$, we get
	\begin{align*}
	\begin{split}
	\xi(D;x,y,z)&=\xi(D_{-f};x,y,z)+y\cdot\xi(D_{/f};x,y,z)+z\cdot\xi(D_{\dagger f};x,y,z)\\
	&=\xi(D_{-f-e};x,y,z)+y\cdot\xi(D_{-f/e};x,y,z)+z\cdot\xi(D_{-f\dagger e};x,y,z)\\
	&\hspace{10pt}+y\cdot\xi(D_{/f};x,y,z)+z\cdot\xi(D_{\dagger f};x,y,z)\\
	&=\xi(D_{-e-f};x,y,z)+y\cdot\xi(D_{/e};x,y,z)\\
	&\hspace{10pt}+z\cdot\xi(D_{\dagger e};x,y,z)+(y+z)\cdot\xi(D_{\dagger f};x,y,z).
	\end{split}
	\end{align*}
	These two expressions are equal.\\
	We check the case 9 similarly:
	\begin{align*}
	\begin{split}
	\xi(D;x,y,z)&=\xi(D_{-e};x,y,z)+y\cdot\xi(D_{/e};x,y,z)+z\cdot\xi(D_{\dagger e};x,y,z)\\
	&=\xi(D_{-e-f};x,y,z)+y\cdot\xi(D_{-e/f};x,y,z)+z\cdot\xi(D_{-e\dagger f};x,y,z)\\
	&\hspace{10pt}+y\cdot\xi(D_{/e-f};x,y,z)+y^2\cdot\xi(D_{/e/f};x,y,z)+yz\cdot\xi(D_{/e\dagger f};x,y,z)\\
	&\hspace{10pt}+z\cdot\xi(D_{\dagger e};x,y,z)\\
	&=\xi(D_{-e-f};x,y,z)+y\cdot\xi(D_{-e/f};x,y,z)+y\cdot\xi(D_{-f/e};x,y,z)\\
	&\hspace{10pt}+(y^2+yz+2z)\cdot\xi(D_{\dagger e};x,y,z),
	\end{split}
	\end{align*}
	and
	\begin{align*}
	\begin{split}
	\xi(D;x,y,z)&=\xi(D_{-f};x,y,z)+y\cdot\xi(D_{/f};x,y,z)+z\cdot\xi(D_{\dagger f};x,y,z)\\
	&=\xi(D_{-f-e};x,y,z)+y\cdot\xi(D_{-f/e};x,y,z)+z\cdot\xi(D_{-f\dagger e};x,y,z)\\
	&\hspace{10pt}+y\cdot\xi(D_{/f-e};x,y,z)+y^2\cdot\xi(D_{/f/e};x,y,z)+yz\cdot\xi(D_{/f\dagger e};x,y,z)\\
	&\hspace{10pt}+z\cdot\xi(D_{\dagger f};x,y,z)\\
	&=\xi(D_{-e-f};x,y,z)+y\cdot\xi(D_{-e/f};x,y,z)+y\cdot\xi(D_{-f/e};x,y,z)\\
	&\hspace{10pt}+(y^2+yz+2z)\cdot\xi(D_{\dagger e};x,y,z),
	\end{split}
	\end{align*}
	we have the same result.\\
	In the case 10 and 11, the arc elimination steps are symmetric in their transformations of $D$ with respect to the order among $e$ and $f$. We have analyzed all of the cases and these complete the proof.
\end{proof}

\begin{definition}
	The \emph{arc elimination polynomial} of a digraph $D$ is defined recursively as follows:
	\begin{align*}
	\begin{split}
	&\xi(D;x,y,z)=\xi(D_{-e};x,y,z)+y\cdot\xi(G_{/e};x,y,z)+z\cdot\xi(G_{\dagger e};x,y,z)\hspace{11pt}\forall e\in E(D),\\
	&\xi(G_1\cup G_2;x,y,z)=\xi(G_1;x,y,z)\cdot\xi(G_2;x,y,z),\\
	&\xi(E_1;x,y,z)=x,\\
	&\xi(E_0;x,y,z)=1.
	\end{split}
	\end{align*}
\end{definition}

The recurrence relation of the trivariate cycle-path polynomial contains a case distinction. Motivated by the co-reduction of the Tutte polynomial 
\begin{align*}
\begin{split}
T(G;x,y)&=\sum_{F\subseteq E(G)}(x-1)^{k(G\langle F\rangle )-k(G)}(y-1)^{|F|+k(G\langle F\rangle )-|V(G)|}\\&=
\begin{cases}
1 & \textrm{if $G$ has no edges,}\\
xT(G_{-e};x,y) & \textrm{if $e$ is a bridge,}\\
yT(G_{/e};x,y) & \textrm{if $e$ is a loop,}\\
T(G_{-e};x,y)+T(G_{/e};x,y) & \textrm{otherwise}
\end{cases}
\end{split}
\end{align*}
and the dichromatic polynomial
\[
Z(G;q,v)=\sum_{F\subseteq E(G)}q^{k(G\langle F\rangle )}v^{|F|}=
\begin{cases}
q^{|V(G)|} & \textrm{if $G$ has no edges,}\\
Z(G_{-e};q,v)+vZ(G_{/e};q,v) & \textrm{for an edge $e$}
\end{cases}
\]
by
\begin{align*}
\begin{split}
&T(G;x,y)=(x-1)^{-k(G)}(y-1)^{-|V(G)|}Z(G;(x-1)(y-1),y-1),\\
&Z(G;q,v)=\left(\frac{q}{v}\right)^{k(G)}v^{|V(G)|}T(G;\frac{q}{v}+1,v+1),
\end{split}
\end{align*}
we pose a question: can we introduce a variable for the initial condition in order to avoid the case distinction, that is, can $\xi(D;x,y,z)$ be determined by $\widehat{\sigma\pi}(D;x,y,z)$ and vice versa? The answer is positive, since the number of vertices after the decomposition contains information about how many arc extraction operations are applied on the loops.

\begin{theorem}
	The arc elimination polynomial $\xi(D;x,y,z)$ and the trivariate cycle-path polynomial $\widehat{\sigma\pi}(D;x,y,z)$ are co-reducible via
	\[
	\widehat{\sigma\pi}(D;x,y,z)=\left(\frac{y-1}{z-1}\right)^{|V(D)|}\xi\left(D;\frac{y-1}{z-1},x\frac{y-1}{z-1},x\frac{(y-1)^2}{z-1}\right)
	\]
	and
	\[
	\xi(D;x,y,z)=x^{|V(D)|}\widehat{\sigma\pi}\left(D;\frac{y}{x},\frac{y+z}{y},\frac{z}{xy}+1\right).
	\]
\end{theorem}
\begin{proof}
	We consider only the arc elimination of a digraph $D$ into empty graphs (at last the disjoint union decomposition may be applied). The result $M$ is a multiset of empty digraphs over $\{E_0,\ldots,E_{|V(D)|}\}$. Since $\xi(D;x,y,z)$ and $\widehat{\sigma\pi}(D;x,y,z)$ are well-defined, the multiset $M$ is independent of the order of arc decomposition. Choose a fixed but arbitrary order of decomposition of $D$ on the arcs into the multiset of empty digraphs $M$. For each $m\in M$, we denote the number of contraction steps on the loops resulting $m$ in this decomposition by $a_l(m)$. Similarly, we denote the number of contraction steps on the non-loop arcs, the number of extraction steps on the loops and non-loop arcs resulting $m$ by $a_2(m)$, $b_1(m)$ and $b_2(m)$, respectively. \\
	Then from the recurrence relation
	\[
	\xi(D;x,y,z)=\xi(D_{-e};x,y,z)+y\cdot\xi(G_{/e};x,y,z)+z\cdot\xi(G_{\dagger e};x,y,z)
	\]
	we have the following expression of $\xi(D;q,v,w)$:
	\[
	\xi(D;q,v,w)=\sum_{m\in M}q^{|V(m)|}v^{a_1(m)+a_2(m)}w^{b_1(m)+b_2(m)}.
	\]
	Since the arc deletion operation has no influence on the vertices, the arc contraction and loop extraction remove one vertex and extraction of a non-loop arc removes two vertices, we have $|V(m)|=|V(D)|-a_1(m)-a_2(m)-b_1(m)-2b_2(m)$ and hence
	\[
	\xi(D;q,v,w)=q^{|V(D)|}\sum_{m\in M}q^{-a_1(m)-a_2(m)-b_1(m)-2b_2(m)}v^{a_1(m)+a_2(m)}w^{b_1(m)+b_2(m)}.
	\]
	Recall that the recurrence relation of $\widehat{\sigma\pi}(D;x,y,z)$ is
	\begin{align*}
		\begin{split}
		&\hspace{10pt}\widehat{\sigma\pi}(D;x,y,z)\\
		&=
		\begin{cases}
		\widehat{\sigma\pi}(D_{-e};x,y,z)+xy\widehat{\sigma\pi}(D_{/e};x,y,z) & \textrm{if $e$ is a loop,}\\
		\widehat{\sigma\pi}(D_{-e};x,y,z)+x\widehat{\sigma\pi}(D_{/e};x,y,z)+x(z-1)\widehat{\sigma\pi}(D_{\dag e};x,y,z) & \textrm{otherwise.}
		\end{cases}
		\end{split}
	\end{align*}

	Since $D_{/e}=D_{\dagger e}$ if $e$ is a loop, we may say
	\[
	\widehat{\sigma\pi}(D;x,y,z)=\widehat{\sigma\pi}(D_{-e};x,y,z)+(xy-\alpha)\widehat{\sigma\pi}(D_{/e};x,y,z)+\alpha\widehat{\sigma\pi}(D_{\dag e};x,y,z),
	\]
	if $e$ is a loop, where $\alpha$ can be chosen arbitrarily. Then we have the following expressions of $\widehat{\sigma\pi}(D)$:
	\[
	\widehat{\sigma\pi}(D;x,y,z)=\sum_{m\in M}(xy-\alpha)^{a_1(m)}\alpha^{b_1(m)}x^{a_2(m)}(x(z-1))^{b_2(m)}.
	\]
	Setting $\alpha=xy-x$, we get
	\[
	\widehat{\sigma\pi}(D;x,y,z)=\sum_{m\in M}(xy-x)^{b_1(m)}x^{a_1(m)+a_2(m)}(x(z-1))^{b_2(m)}.
	\]
	Equation $\xi(D;q,v,w)=q^{|V(D)|}\widehat{\sigma\pi}(D;x,y,z)$ holds, if
	\begin{multline*}
		q^{-a_1(m)-a_2(m)-b_1(m)-2b_2(m)}v^{a_1(m)+a_2(m)}w^{b_1(m)+b_2(m)}\\=(xy-x)^{b_1(m)}x^{a_1(m)+a_2(m)}(x(z-1))^{b_2(m)},
	\end{multline*}
	that is,
	\begin{multline*}
		\left(\frac{v}{q}\right)^{a_1(m)+a_2(m)}\left(\frac{w}{q^2}\right)^{b_1(m)+b_2(m)}q^{b_1(m)}\\=x^{a_1(m)+a_2(m)}(x(z-1))^{b_1(m)+b_2(m)}\left(\frac{y-1}{z-1}\right)^{b_1(m)}.
	\end{multline*}
	
	Applying ``equating exponents", we conclude that 
	\[
	x=\frac{v}{q},\,\,\,y=\frac{v+w}{v},\,\,\, z=\frac{w}{qv}+1
	\]
	or
	\[q=\frac{y-1}{z-1},\,\,\,v=x\frac{y-1}{z-1},\,\,\,w=x\frac{(y-1)^2}{z-1},\]
	this completes the proof.
\end{proof}

We have now an interest in the combinatorial interpretation of the coefficients of $\xi(D;x,y,z)$. In the next theorem, an explicit expression of the arc elimination polynomial is given.
\begin{theorem}
	\[
	\xi(D;x,y,z)=\sum_{A,B}x^{k(D\langle A\cup B\rangle)-c(D\langle B\rangle )-c_1(D\langle A\rangle)}y^{|A|+|B|-c(D\langle B\rangle )}z^{c(D\langle B\rangle)},
	\]
	where the sum is over all subsets $A,B\subseteq E(D)$ of $E(D)$ such that
	\begin{enumerate}\setlength{\itemsep}{0pt}
		\item $A\cap B=\emptyset$,
		\item there is no vertex such that an arc in $A$ and an arc in $B$ are incident to it, and
		\item each component of the spanning subgraph $D\langle A\cup B\rangle$ is either a cycle or a path or an isolated vertex.
	\end{enumerate}
	Here $k(D)$ denotes the number of components of $D$, $c(D)$ denotes the number of covered components of $D$, that is, components of $D$ which are not isolated vertices, and $c_1(D)$ denotes the number of cycles of length 1 (loops) in $D$.
\end{theorem}
\begin{proof}
	Let $D=(V,E)$ be a (multi-)digraph. The set of pairs $(A,B)$ of arc subsets $A,B\subseteq E$ satisfying the three conditions in the theorem is denoted by $\mathcal{C}(D)$. Let $N(D)$ be defined explicitly as
	\[
	N(D;x,y,z):=\sum_{(A,B)\in \mathcal{C}(D)}x^{k(D\langle A\cup B\rangle)-c(D\langle B\rangle )-c_1(D\langle A\rangle)}y^{|A|+|B|-c(D\langle B\rangle )}z^{c(D\langle B\rangle)}.
	\]
	We may use the notation 
	\[f(D,(A,B)):=x^{k(D\langle A\cup B\rangle)-c(D\langle B\rangle )-c_1(D\langle A\rangle)}y^{|A|+|B|-c(D\langle B\rangle )}z^{c(D\langle B\rangle)},\] then $N(D;x,y,z):=\sum_{(A,B)\in \mathcal{C}(D)}f(D,(A,B))$.\\ 
	In order to proof $\xi(D;x,y,z)=N(D;x,y,z)$, we need to show that $N(D)$ satisfies
	\begin{align*}
	\begin{split}
	&N(D;x,y,z)=N(D_{-e};x,y,z)+y\cdot N(G_{/e};x,y,z)+z\cdot N(G_{\dagger e};x,y,z)\hspace{11pt}\forall e\in E,\\
	&N(E_n;x,y,z)=x^n.
	\end{split}
	\end{align*}
	For the empty digraph $E_n$, the only summand corresponds to $A=B=\emptyset$, and obviously $N(E_n;x,y,z)=x^n=\xi(E_n;x,y,z)$.\\
	Let $e\in E$ be an arbitrarily chosen arc. The summands can be divided into three disjoint cases:
	\begin{itemize}\setlength{\itemsep}{0pt}
		\item Case 1: $e\notin A\cup B$;
		\item Case 2: $e\in B$ and $e$ is the only arc of a component of $D\langle B\rangle$;
		\item Case 3: All the rest. That is, $e\in A$ or $e\in B$ but it is not the only arc of a component of $D\langle B\rangle$.
	\end{itemize}
	The sets of arc subset pairs $(A,B)\in \mathcal{C}(D)$ satisfying the conditions in case 1, 2 and 3 are denoted by $\mathcal{C}_1(D)$, $\mathcal{C}_2(D)$ and $\mathcal{C}_3(D)$, respectively. \\
	In the case 1, it is easily to seen that $\mathcal{C}_1(D)=\mathcal{C}(D_{-e})$. Then
	\[
	N_1(D):=\!\!\!\sum_{(A,B)\in \mathcal{C}_1(D)}\!\!\!x^{k(D\langle A\cup B\rangle)-c(D\langle B\rangle )-c_1(D\langle A\rangle)}y^{|A|+|B|-c(D\langle B\rangle )}z^{c(D\langle B\rangle)}=N(D_{-e}).
	\]
	In the case 2, $e\in B$ is the only arc of a component of $D\langle B\rangle$, because of the required condition, any arc incident to $e$ can not in $A$ or $B$. Thus we can define a bijection $\varphi: \mathcal{C}_2(D)\rightarrow \mathcal{C}(D_{\dagger e})$, $\varphi((A,B)):=(A,B\backslash\{e\})$. Now compare $D_{\dagger e}$ with $D$, we get
	\begin{align*}
	\begin{split}
	&|B\backslash\{e\}|=|B|-1,\\ 
	&k(D_{\dagger e}\langle A\cup B\backslash\{e\}\rangle)=k(D\langle A\cup B\rangle)-1, \textrm{ and}\\ 
	&c(D_{\dagger e}\langle B\backslash\{e\}\rangle ) =c(D\langle B\rangle )-1.
	\end{split}
	\end{align*}
	that is,
	\[
	f(D,(A,B))=z\cdot f(D_{\dagger e},\varphi((A,B))) \hspace{0.5cm}\forall (A,B)\in\mathcal{C}_2(D)
	\]
	and therefore,
	\begin{align*}
	\begin{split}
	N_2(D):=&\sum_{(A,B)\in \mathcal{C}_2(D)}x^{k(D\langle A\cup B\rangle)-c(D\langle B\rangle )-c_1(D\langle A\rangle)}y^{|A|+|B|-c(D\langle B\rangle )}z^{c(D\langle B\rangle)}\\
	=&\sum_{(A,B)\in \mathcal{C}_2(D)}f(D,(A,B))\\
	=&\,\,\,\,z\cdot\sum_{(A,B)\in \mathcal{C}_2(D)}f(D_{\dagger e},\varphi((A,B)))\\
	=&\,\,\,\,z\cdot\sum_{(A,B)\in \mathcal{C}(D_{\dagger e})}f(D_{\dagger e},(A,B))\\
	=&\,\,\,\,z\cdot N(D_{\dagger e}).
	\end{split}
	\end{align*}
	In the case 3, either $e\in A$ or $e\in B$ and $e$ is incident to other arcs in $B$. Since $e$ is either the only arc of a component of $D\langle A\rangle$, or belongs to a directed path or a directed cycle of length at least two, whose arcs are either all in $A$ or all in $B$, we can define a function $\psi: \mathcal{C}_3(D)\rightarrow\mathcal{C}(D_{/e})$, $\psi((A,B)):=(A\backslash\{e\},B\backslash\{e\})$. Evidently
	\[
	\psi^{-1}((A,B)):=
	\begin{cases}
	(A,B\cup\{e\}) & \textrm{if $e$ is incident to an arc of $B$,}\\
	(A\cup\{e\},B) & \textrm{otherwise}
	\end{cases}
	\]
	is the inverse function of $\psi$, and the well-definedness of $\psi^{-1}$ is guaranteed by the conditions of $(A,B)$, we conclude that $\psi$ is bijective. 
	Compare now $D_{/e}$ with $D$, we get
	\begin{align*}
	\begin{split}
	&|A\backslash\{e\}|+|B\backslash\{e\}|=|A|+|B|-1,\\ 
	&c(D_{/ e}\langle B\backslash\{e\}\rangle ) =c(D\langle B\rangle ),\\
	&k(D_{/ e}\langle A\cup B\backslash\{e\}\rangle)=
	\begin{cases}
	k(D\langle A\cup B\rangle)-1 & \textrm{if $e\in A$ is a loop,}\\
	k(D\langle A\cup B\rangle) &\textrm{otherwise,}
	\end{cases}\\
	&c_1(D_{/e}\langle A\rangle)=
	\begin{cases}
	c_1(D\langle A\rangle)-1 & \textrm{if $e\in A$ is a loop,}\\
	c_1(D\langle A\rangle) & \textrm{otherwise.}
	\end{cases}
	\end{split}
	\end{align*}
	Applying to the function $f$, we have
	\[
	f(D,(A,B))=y\cdot f(D_{/ e},\psi((A,B))) \hspace{0.5cm}\forall (A,B)\in\mathcal{C}_3(D).
	\]
	Therefore,
	\begin{align*}
	\begin{split}
	N_3(D):=&\sum_{(A,B)\in \mathcal{C}_3(D)}x^{k(D\langle A\cup B\rangle)-c(D\langle B\rangle )-c_1(D\langle A\rangle)}y^{|A|+|B|-c(D\langle B\rangle )}z^{c(D\langle B\rangle)}\\
	=&\sum_{(A,B)\in \mathcal{C}_3(D)}f(D,(A,B))\\
	=&\,\,\,\,y\cdot\sum_{(A,B)\in \mathcal{C}_3(D)}f(D_{\dagger e},\psi((A,B)))\\
	=&\,\,\,\,y\cdot\sum_{(A,B)\in \mathcal{C}(D_{/ e})}f(D_{/ e},(A,B))\\
	=&\,\,\,\,y\cdot N(D_{/ e}).
	\end{split}
	\end{align*}
	Summing up the three cases, we conclude that
	\[
	N(D)=N_1(D)+N_2(D)+N_3(D)=N(D_{-e})+y\cdot N(D_{/e})+z\cdot N(D_{\dagger e}).
	\]
	Together with $N(E_n)=x^n$ it implies $N(D)=\xi(D)$. This completes the proof.
\end{proof}
\section*{Acknowledgement}
The author is grateful to Professor Peter Tittmann for the motivation and comments on this paper.


\begin{thebibliography}{99}
	\bibitem{averbouch2011completeness}
	Ilia Averbouch.
	\newblock Completeness and universality properties of graph invariants and
	graph polynomials.
	\newblock 2011.
	
	\bibitem{averbouch2008most}
	Ilia Averbouch, Benny Godlin, and Johann~A Makowsky.
	\newblock A most general edge elimination polynomial.
	\newblock In {\em International Workshop on Graph-Theoretic Concepts in
		Computer Science}, pages 31--42. Springer, 2008.
	
	\bibitem{averbouch2010extension}
	Ilia Averbouch, Benny Godlin, and Johann~A Makowsky.
	\newblock An extension of the bivariate chromatic polynomial.
	\newblock {\em European Journal of Combinatorics}, 31(1):1--17, 2010.
	
	\bibitem{chung2016matrix}
	Fan Chung and Ron Graham.
	\newblock The matrix cover polynomial.
	\newblock {\em Journal of Combinatorics}, 7(2):375--412, 2016.
	
	\bibitem{chung1995cover}
	Fan~RK Chung and Ronald~L Graham.
	\newblock On the cover polynomial of a digraph.
	\newblock {\em Journal of Combinatorial Theory, Series B}, 65(2):273--290,
	1995.
	
	\bibitem{d2000cycle}
	Ottavio~M D'Antona and Emanuele Munarini.
	\newblock The cycle-path indicator polynomial of a digraph.
	\newblock {\em Advances in Applied Mathematics}, 25(1):41--56, 2000.
	
	\bibitem{dohmen2003new}
	Klaus Dohmen, Andr{\'e} P{\"o}nitz, and Peter Tittmann.
	\newblock A new two-variable generalization of the chromatic polynomial.
	\newblock {\em Discrete Mathematics \& Theoretical Computer Science},
	6(1):69--90, 2003.
	
	\bibitem{farrell1979introduction}
	Edward~J Farrell.
	\newblock An introduction to matching polynomials.
	\newblock {\em Journal of Combinatorial Theory, Series B}, 27(1):75--86, 1979.
	
	
	\bibitem{tittmann2011enumeration}
	Peter Tittmann, Ilya Averbouch, and Johann~A Makowsky.
	\newblock The enumeration of vertex induced subgraphs with respect to the
	number of components.
	\newblock {\em European Journal of Combinatorics}, 32(7):954--974, 2011.
	
	\bibitem{tutte1954contribution}
	William~T Tutte.
	\newblock A contribution to the theory of chromatic polynomials.
	\newblock {\em Canad. J. Math}, 6(80-91):3--4, 1954.

\end{thebibliography}
\end{document}